\documentclass[letterpaper, 10 pt, conference]{ieeeconf}  

\IEEEoverridecommandlockouts                              

\usepackage[T1]{fontenc}	
\usepackage[utf8]{inputenc}
\usepackage{graphics} 

\let\labelindent\relax
\let\labelindent\relax 
\usepackage{amsmath,mathtools,amsfonts,amsthm} 
\usepackage{amssymb}
\usepackage{accents}
\usepackage{verbatim}
\usepackage{enumitem,kantlipsum}
\usepackage{color}
\usepackage{xcolor}
\usepackage{cite}
\usepackage{algorithm}
\usepackage[noend]{algpseudocode}

\newtheorem{thm}{Theorem}

\newtheorem{cor}{Corollary}

\newtheorem{prob}{Problem}

\newtheorem{rem}{Remark}
\newtheorem{ass}{Assumption}
\title{\LARGE \bf
Prediction-Based Leader-Follower Rendezvous Model Predictive Control with Robustness to Communication Losses
}

\author{Dženan Lapandić$^{1}$, Christos K. Verginis$^{2}$, Dimos V. Dimarogonas$^{1}$ and Bo Wahlberg$^{1}$
\thanks{This work was supported by the Wallenberg AI, Autonomous Systems and Software Program (WASP) and the Swedish Research Council, Knut and Alice Wallenberg Foundation (KAW).}
\thanks{$^{1}$Dženan Lapandić, Dimos V. Dimarogonas and Bo Wahlberg are with  Division of Decision and Control Systems, KTH Royal Institute of Technology, Stockholm, Sweden.   {\tt\small lapandic,dimos,bo@kth.se}}%
\thanks{$^{2}$Christos K. Verginis is with Division of Signals and Systems, Department of Electrical Engineering, Uppsala University,
Uppsala, Sweden. {\tt\small christos.verginis@angstrom.uu.se}}}

\begin{document}

\maketitle
\thispagestyle{empty}
\pagestyle{empty}

\begin{abstract}
In this paper we propose a novel distributed model predictive control (DMPC) based algorithm with a trajectory predictor for a scenario of landing of unmanned aerial vehicles (UAVs) on a moving unmanned surface vehicle (USV). 
The algorithm is executing DMPC with exchange of trajectories between the agents at a sufficient rate. 
In the case of loss of communication, and given the sensor setup, agents are predicting the trajectories of other agents based on the available measurements and prior information.
The predictions are then used as the reference inputs to DMPC. 
During the landing, the followers are tasked with avoidance of USV-dependent obstacles and inter-agent collisions. 
In the proposed distributed algorithm, all agents solve their local optimization problem in parallel and we prove the convergence of the proposed algorithm. 
Finally, the simulation results support the theoretical findings.
\end{abstract}

\section{INTRODUCTION}
The problem of autonomous navigation and landing of unmanned aerial vehicles (UAVs) on other autonomous agents has drawn a considerable attention in recent years. 
The coordination and control of UAVs with optimization-based methods can be challenging in such scenarios due to the limited computational resources on-board. 
In that respect, a relevant methodology is cooperative distributed model predictive control whose aim is to optimize control inputs of a considered agent given the predicted trajectories of other (neighbouring) agents over the planning horizon. 
Each agent solves the finite-horizon distributed optimal control problem (DOCP), applies its control input and broadcasts the predicted trajectory to other agents \cite{muller2012cooperative,christofides2013distributed,verginis2018communication}. 

The considered agents in our setup are equipped with sensors, i.e. a camera on the quadrotor (UAV) side, and a radar on the boat (USV) side, that enable the accurate position measurement of the neighbouring agents with a frequency significantly higher than the real-time execution frequency of the MPC. 
In the case of loss of communication between the agents, an agent is left only with the data history and locally available measurements. 
During challenging maneuvers such as landing or navigating through a space with moving obstacles, the frequency of execution of the DMPC relative to the moving speeds of the agents and obstacles can severely affect the performance and safety. 
This is especially the case when the references are based only on the current measurements and not on the committed trajectories, which indeed occurs in the loss of communication scenario. 
Note that, the communication losses we consider can be both temporary and permanent, because even an occasional package drop for few seconds that can occur during the hardware-in-the-loop experiments can lead to disastrous outcomes.

The hierarchical approaches presented in the literature \cite{lopez2018robust,csomay2022multi} deal with the described problem by separating the control in different layers. 
Usually, there is a high-level planning layer and a safety layer, for example with Control Barrier Functions (CBF) operating on a much higher frequency that is tasked with safety in the presence of obstacles \cite{rosolia2022unified}. 
Fundamentally, the concept of MPC is to unify these layers. 
Both approaches have their advantages and disadvantages and all aspects of the chosen concept must be taken into account to avoid unwanted behaviors. 
One of the ways is to use a prediction scheme to complement the available data. In \cite{falanga2017vision}, the authors propose a vision-based framework with the state estimation for the ground vehicle considered as a moving target. Estimation-based control protocol proposed in \cite{song2022distributed} uses only local observations of the state of the neighboring agents for rendezvous and flocking control. In another estimation-based protocol, authors in \cite{lymperopoulos2008adaptive} use particle filter to predict aircraft trajectories. Prediction-based navigation in a decentralized event-based scheme is studied in \cite{maniatopoulos2012decentralized}. Moreover, learning-based prediction approaches for multi-agent systems in the recent years became very popular area of research \cite{felsen2018will, li2019interaction, zhao2019multi,li2020evolvegraph, cao2022advdo}. Unfortunately, most of these approaches require extensive amount of data for the training purpose and direct transferability to our particular application is unknown.

In this paper, we case the multi-agent and heterogeneous (involving UAVs and USVs) rendezvous problem as a leader-follower network with one leader and one or more follower agents for autonomous landings. 
The leader in our case is the USV boat and UAVs are considered as followers. 
We formulate the problem as a cooperative distributed model predictive control problem with collision avoidance constraints. 
In the case when the predicted trajectory is available to the followers, and the follower is able to dynamically follow the given reference trajectory, the landing can be executed using a relatively simple MPC for a single agent follower. 
However, a sudden communication loss can severely destabilize the landing leading to a collision with obstacles on the boat. 
Based on this observation, we adopt a prediction scheme in the case of communication loss to enhance the safety and performance.

Moreover, we generalize the scheme to multiple-follower rendezvous in which the inter-agent collisions must be handled. 
There are several challenges in the inter-agent collision avoidance in this case. 
First, in the sequential application of the algorithm each agent is optimizing its control strategy based on the shared trajectories that other agents have committed to. 
In this case, each agent must wait until all other agents have shared their new trajectories to begin the computation of its next control input. 
Otherwise, using the old shared trajectories without additional safety measures can lead to a collision. 
Second, in the parallel version the agents operate with the most recent available shared trajectories and account for the worst-case deviations as a safety measure. 
However, this can be very conservative and lead to deadlock.
Third, if one or more agents lose the communication for some duration safety can be compromised. 
This, in some sense, resembles a non-cooperative scenario and we propose appropriate safety measures.

The contributions of this paper are 
\begin{itemize}
    \item a rendezvous algorithm based on leader-follower DMPC formulation for autonomous landing in presence of obstacles
    \item a convergence proof for a probabilistic safe landing in the case of communication loss
    \item an inter-agent collision avoidance robust to communication losses 
\end{itemize}

The paper is organized as follows. The problem formulation is given in Sec.~\ref{sec:prob_form} and the control architecture is described in Sec.~\ref{sec:ctrl_arch}. In Sections~\ref{sec:convergence} and ~\ref{sec:rendezvous_algorithm}, one-follower and multiple-follower cases are examined and the proposed rendezvous algorithm is stated. Finally, Sec.~\ref{sec:sim_results} presents the simulation results and Sec.~\ref{sec:conclusion} concludes the paper.
\subsection{Notation}
We denote the discrete time step with $t$ and set of agents with $\mathcal{N} := \{ l,f_1,...,f_M \}$.
The state trajectories are denoted with $x_i(t)$, the predicted nominal state trajectories with $\hat{x}_i(k|t)$ and optimal state trajectories with $\hat{x}_i^*(k|t)$ for an agent $i$ at time step $k+t$ predicted at time step $t$. The collection of trajectories from time step $t$ until $t+N$ is denoted with $x_i(\cdot|t)$; $\mathcal{B}_{r}:=\{ a: \| a \| \leq r\}$ is a ball of radius $r$; $ \oplus $ is Minkowski sum addition defined as $x \oplus S:=\{ x + a: a \in S\}$. 
For an element $i\in\mathcal{N}$, we denote a set excluding the element $i$ with $\mathcal{N}_{-i}:= \mathcal{N} \setminus \{i\}$.

\section{PROBLEM FORMULATION}\label{sec:prob_form}

Consider a multi-agent system consisting of one leader and $M$ follower agents that are dynamically decoupled and behaving according to the nonlinear discrete time dynamics
\begin{equation}\label{eq:multi-agent_system}
x_i(t+1)=f_i(x_i(t),u_i(t)) + w_i(t),
\end{equation}
where $i \in \mathcal{N} = \{ l,f_1,...,f_M \}$, $l$ denotes the leader, $f_1,...,f_M$ denote the followers, $M$ is the number of follower agents, $x_i \in \mathcal{X}_i \subseteq \mathbb{R}^{n_i}$, $u_i \in \mathcal{U}_i\subseteq \mathbb{R}^{m_i}$ denote the state and input of an agent $i$ that are subject to the state and input constraints $\mathcal{X}_i$ and $\mathcal{U}_i$, respectively, and $w_i \in \mathcal{W}_i \subset \mathbb{R}^{n_i}$ is unknown but bounded disturbance in a compact set $\mathcal{W}_i$. 

The follower agents have the same dynamics and state space of a quadrotor UAV as in \cite{lapandic2022robust} which are different from the leader ones modeled as a 3DoF boat USV model \cite{fossen2011handbook}. We assume that the landing platform is rigidly attached to the boat USV and neglect the heave (vertical motion), roll and pitch motion.  Furthermore, the first three elements of state vectors of all agents denote the position  $p_i\in\mathbb{R}^3$. 

The overall multi-agent system dynamics in stack-vector form are 
\begin{equation}\label{eq:mas_stacked}
    x(t+1) = f(x(t),u(t)) + w(t),
\end{equation}
where $x=[x_l^T,x_{f_1}^T,...,x_{f_M}^T]^T$, $u=[u_l^T,u_{f_1}^T,...,u_{f_M}^T]^T$, $w=[w_l^T,w_{f_1}^T,...,w_{f_M}^T]^T$. The state and input constraints of the overall system are $x \in \mathcal{X} := \mathcal{X}_l \times \mathcal{X}_{f_1} \times ... \times \mathcal{X}_{f_M} $, and $u \in \mathcal{U} := \mathcal{U}_l \times \mathcal{U}_{f_1} \times ... \times \mathcal{U}_{f_M} $, the set of all disturbances is $w \in \mathcal{W} := \mathcal{W}_l \times \mathcal{W}_{f_1} \times ... \times \mathcal{W}_{f_M}$.

We formulate the inter-agent collision avoidance constraints as
\begin{equation}\label{eq:collision_formal_formulation}
    h_{ij}(x_i(t),x_j(t)) \geq 0, \quad \forall i,j\in \mathcal{N}, i \neq j,
\end{equation}
where $h_{ij}:\mathcal{X}_i \times \mathcal{X}_j \rightarrow \mathbb{R}$ is a function that encodes collisions and will be defined later. Given a solution $x(t)$ of the system \eqref{eq:mas_stacked}, if the constraint functions \eqref{eq:collision_formal_formulation} are satisfied for all $t\geq t_0$ then $x(t)$ is a collision-free solution, i.e. $x(t)\in \mathcal{F} \subseteq \mathcal{X}$, where $\mathcal{F}$ denotes the collision-free space.

This paper considers the problem of navigating the leader agent to follow a given reference and follower agents to a rendezvous position with respect to the leader while avoiding inter-agent collisions. 
We denote the leader reference $x_{r,l}(t)=x_r(t)$ and references of the followers $x_{r,i}(x_l(t),c_i):=[(p_l(t) + c_i)^T,0_{n_i-3}^T]^T$ with respect to the position $p_l(t)$ of the leader $x_l$ and a given offset $c_i \in \mathbb{R}^3$ encoding the particular landing position of an $i$-th agent, $i\in \mathcal{N}_{-l}$. With a slight abuse of notation, the problem treated in the paper can be stated as follows.
\begin{prob}\label{prob:formulation}
    Consider a multi-agent system \eqref{eq:multi-agent_system}. Design a control policy $u$ such that
    \begin{align*}
        &\phantom{x}x(t) \in \mathcal{F}, \quad t\geq t_0 \\
        &\lim_{t\rightarrow \infty} (x_l(t),u_l(t)) = (x_r(t),0_{n_l}) \\
        &\lim_{t\rightarrow \infty} (x_i(t),u_i(t)) = (x_{r,i}(x_l(t),c_i),0_{n_i}), \forall i \in \mathcal{N}_{-l}
    \end{align*}
\end{prob}

Moreover, we assume that agents are able to communicate and share their current and predicted positions $\hat{z}_i(\cdot|t)$ asynchronously. 
Because the follower agents have different dynamics from the leader, and different state vectors, we define the following mapping from the leader state space $\mathcal{X}_l$ to the follower state space $\mathcal{X}_i$, $i\in \mathcal{N}_{-l}$
\begin{equation}
\hat{z}_{l}(t) = H\hat{x}_l(t), \label{eq:ch6_measurement_z} 
\end{equation} 
where $H = \text{diag}(I_3,0_{(n_i - 3)\times (n_l -3)})$. This mapping effectively maps only the position of the leader to the follower state space. For the follower agents it holds $\hat{z}_{i}(t) = \hat{x}_i(t)$, $i\in \mathcal{N}_{-l}$.

When the communication is without losses, all follower agents have access to $\hat{z}_i(\cdot|t-1)$ for all $i\in \mathcal{N}$, where all data arrived between time steps $t-1$ and $t$ is cast as data from time step $t-1$.
However, we want that the designed control policy $u$ is able to achieve the goal stated in Problem~\ref{prob:formulation} even in the case of communication loss.
In the context of this paper, the communication loss is considered as an inability of an agent to retrieve the latest shared data from another agent. More formally, the communication loss is an inability of an agent $i$ at time step $t$ to retrieve the shared data $\hat{z}_j(\cdot|t-k)$ from another agent $j$ generated at time step $t-k$, where $k \geq 1$.
Thus, we will impose stronger assumptions on the behaviour of the followers that will be stated in Sections~\ref{sec:rendezvous_algorithm} and \ref{sec:convergence}.

\section{CONTROL ARCHITECTURE}\label{sec:ctrl_arch}
\subsection{Distributed MPC Formulation}
We choose to address Problem~\ref{prob:formulation} using distributed MPC where each agent solves a distributed optimal control problem for a planning horizon of $N$ time steps and applies the first control input during the control horizon $\Delta t$ \cite{muller2012cooperative}. Note that all time steps in the control architecture are discrete $t\in \mathbb{N}_0$ and $\Delta t$ denotes the continuous time duration between two time steps. 
The procedure is then repeated at every time step for each agent.

Let $\hat{x}_i(k|t)$ be the nominal state trajectory at time $k+t$, $k=0,1,...,N$ calculated at time instant $t$, where $\hat{x}_i(0|t)=x_i(t)$, governed by the following difference equation
\begin{equation}
\hat{x}_i(k+1|t)=f_i(\hat{x}_i(k|t),u_i(k|t)) \label{eq:nominal_dynamics}
\end{equation}
Each agent is provided with the state reference trajectory at time step $t$ until $t+N$ which is given as $x_{r,i}(\cdot|t) = \{x_{r,i}(0|t),...,x_{r,i}(N|t)\}$. The references are provided by the proposed algorithm which will be elaborated in Sections~\ref{sec:convergence} and \ref{sec:rendezvous_algorithm}. The control objective of each agent at time step $t$ is to minimize the following cost function
\begin{equation}\label{eq:cost}
J_i(\hat{x}_i(\cdot|t),u_i(\cdot|t),x_{r,i}(\cdot|t),N,t) = \sum_{k=0}^{N}  \left \| \hat{x}_i(k|t) - x_{r,i}(k|t)  \right \|_{Q_i}^2 
\end{equation}%
 while respecting the constraints, where the summand represents the stage cost and $Q_i$ is a positive-definite weighting matrix.  
We denote also $J=\sum_{i\in \mathcal{N}} J_{i}$.

We formulate the distributed optimal control problem with respect to the objective.
\begin{prob}\label{prob:ch6_docp_standard}
Let the states of the agents at time $t$ be $x_i(t)$, $i\in \mathcal{N}$. Given the references $x_{r,i}(\cdot|t)$ and the predicted trajectories of other agents $\hat{z}_j(\cdot|t)$, $j\in \mathcal{N}_{-i}$, the distributed optimal control problem is formulated as 
\begin{subequations}\label{eq:ch6_mpc_problem}
\begin{equation}
    \min_{u_i(\cdot|t)}
J_i(\hat{x}_i(\cdot|t),u_i(\cdot|t),x_{r,i}(\cdot|t),N,t)    \label{eq:ch6_8} 
\end{equation}
subject to
\begin{align}
    &    \hat{x}_i(k+1|t)=f_i(\hat{x}_i(k|t),u_i(k|t)),\label{eq:ch6_p_1} \\
    &\hat{x}_i(k|t) \in \mathcal{X}_i,\label{eq:ch6_p_2} \\
    &u_i(k|t) \in \mathcal{U}_i,\label{eq:ch6_p_3} \\
    &\hat{x}_i(k|t) \in \mathcal{X}_{i,j}(\hat{z}_j(k|t)), \quad \text{for } i \in \mathcal{N}_{-l}, j \in \mathcal{N}_{-i},\label{eq:ch6_p_4}
\end{align}
\end{subequations}
for $k=0,1,...,N$, for all $i\in \mathcal{N}$.
\end{prob}

Set $\mathcal{X}_i$ denotes the set of model state constraints and $\mathcal{U}_i$ the input constraints. 
$\mathcal{X}_{i,j}(\hat{z}_j(k|t))$ is the set of spatiotemporal safety constraints with respect to the other agents that will be formulated in the next section. Note that the leader is not subjected to the inter-agent collision constraints \eqref{eq:ch6_p_4}, but only the follower agents.

\begin{figure}[t]
    \centering
      \raisebox{-0.5\height}{
        \includegraphics[width=0.47\linewidth]{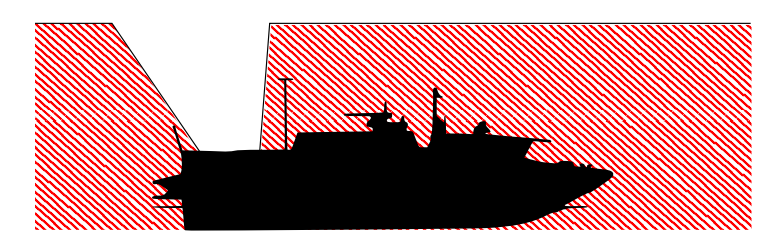}}
        \hfill
    \raisebox{-0.42\height}{      \includegraphics[width=0.47\linewidth]{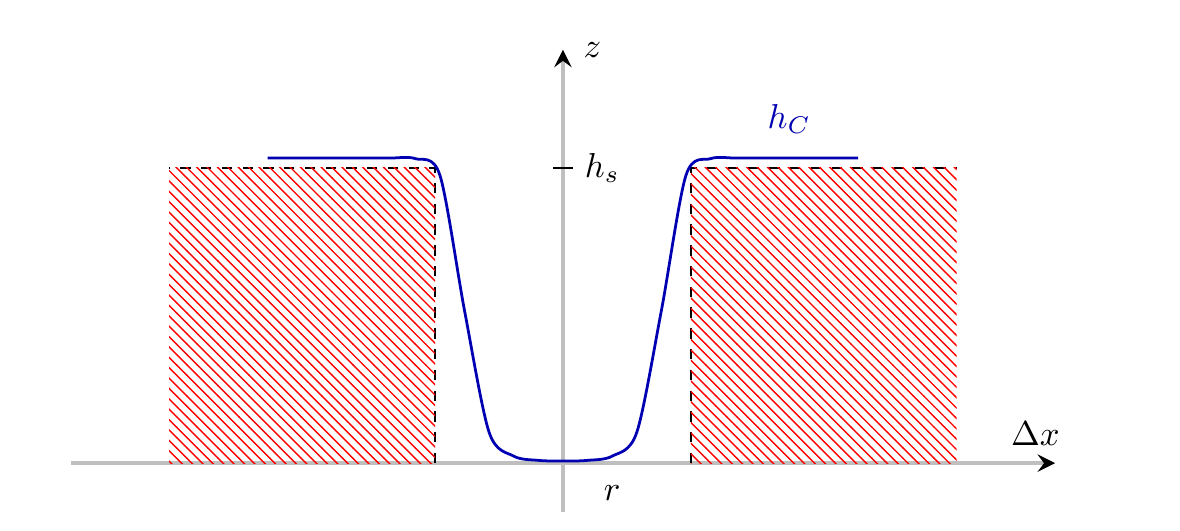}}
    \caption{Restricted area and side view of the constraint $h_C$ from \eqref{eq:ch6_nonlinear_cstr} }
    \label{fig:ch6_restricted_flying_area}
\end{figure}

\subsection{Spatiotemporal Safety Constraints}
The desired specification on UAV movements in the 3D space are to fly above some prescribed height and avoid obstacles.
However, in the problem of UAV landing on a boat, we are particularly interested in the obstacles that arise from the boat shape, equipment and its movement. 
These specific obstacles surrounding the landing area on the boat are depicted in Figure~\ref{fig:ch6_restricted_flying_area}.

This area can be modeled as two convex areas with a binary variable that depends on the altitude of the UAV and determines which of the two constraints should be enforced within the mixed-integer optimization problem \cite{persson2019model}. 
However, it is of interest to model the restricted flying area with one nonlinear continuous function such that the usage of mixed-integer programming is avoided due to the computational burden it can cause. 

In the following considerations, we assume the landing platform is circular with the radius $r$, $r_{safe} < r$ is the safety radius needed for safe landing and $h_s$ is the safety height below which the UAV is not allowed to descent unless above the landing platform. Let $z_l = [p_{x_c},p_{y_c},0_{n_f-2}^T]^T\in\mathcal{X}_f$ be the center of the boat landing platform and $x_f=[p_x,p_y,p_z,0_{n_f-3}^T]^T\in\mathcal{X}_f$ be the position of the UAV.

We define the constraint function $h_{C}$ in a novel way as 
\begin{equation}\label{eq:ch6_nonlinear_cstr}
    h_{C}(x_f,z_l):= p_z - \frac{h_s}{1+e^{-\beta ((p_x-p_{x_c})^2 + (p_y-p_{y_c})^2 - r^2)}} \geq 0
\end{equation}
where $\beta>0$ is a tuning parameter that determines the slope of the funnel. 

Requiring that $h_{C}(x_f,z_l) \geq 0$, the UAV will always be above the boat position-dependent constraint function. The boundary of the restricted area defined with $h_{C}$ in 3D is visible on Fig.~\ref{fig:ch6_cstr}.
\begin{figure}
        \centering
        \includegraphics[width=0.5\linewidth]{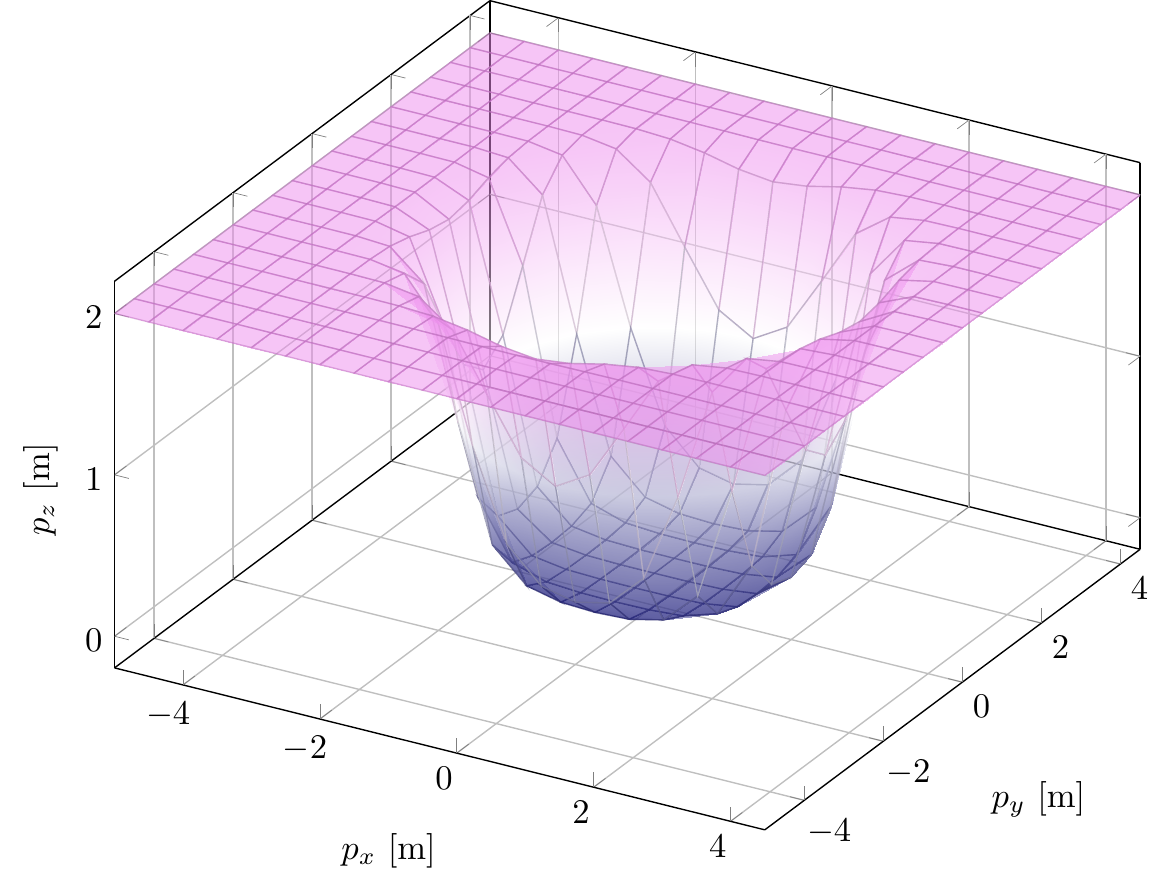} 
        \caption{$h_{C}$ from \eqref{eq:ch6_nonlinear_cstr}, for $h_s=2,r=2.5,\beta=1$}\label{fig:ch6_cstr}
    \vspace{-0.7cm}
\end{figure}

The spatiotemporal constraint imposed on the position of the follower $\hat{x}_{i}(k|t)$, $i\in \mathcal{N}_{-l}$ at time step $t+k$, predicted at time $t$, for $i=0,1,...,N$, with respect to the future trajectory of the leader position $\hat{z}_l(k|t)$ is
\begin{equation}
 h_C(\hat{x}_{i}(k|t),\hat{z}_l(k|t))\geq 0 \quad \text{for all } k=0,...,N.
\end{equation}
Therefore, we define the safe set $\mathcal{X}_{i,l}(\hat{z}_l(k|t))$ in \eqref{eq:ch6_p_4} for all $i\in \mathcal{N}_{-l}$ as 
\begin{equation}
\mathcal{X}_{i,l}(\hat{z}_l(k|t)) := \left \{\hat{x}_{i}(k|t) :  h_C(\hat{x}_{i}(k|t),\hat{z}_l(k|t))\geq 0 \right \}.
\end{equation}
Note that in the current formulation we assume that the height of the landing platform is zero.

\subsection{Inter-agent Collision Avoidance}
The previous section defines the safety constraint function that each follower enforces in its optimization problem with respect to the leader. The inter-agent collision avoidance condition is formulated on a set of follower agents $\mathcal{M}=\mathcal{N}_{-l}$ as 
\begin{equation}
\left \| C(\hat{x}_i(k|t) - \hat{z}_j(k|t)) \right \| \geq R, \text{  for } j\in\mathcal{M}_{-i}\label{eq:ch6_coll_avoid}
\end{equation}
for $k=0,1,...,N$, where $R>0$ is the minimal distance between the follower agents. The matrix $C=\text{diag}(1,1,c,0_{n_i-3})$, with $c>1$, determines the shape of the super ellipsoidal set constraint such that it prevents the collision and the downwash effect that can occur if a UAV ends up below another.

Moreover, the condition \eqref{eq:ch6_coll_avoid} can be stated in form \eqref{eq:collision_formal_formulation} as 
\begin{equation*}
    h_{ij}(\hat{x}_{i}(k|t),\hat{z}_j(k|t)):= \left \| C(\hat{x}_i(k|t) - \hat{z}_j(k|t)) \right \| - R \geq 0
\end{equation*}
and the safe set $\mathcal{X}_{i,j}(\hat{z}_j(k|t))$ in \eqref{eq:ch6_p_4} for all $i\in \mathcal{M}$ and $j \in \mathcal{M}_{-i}$ is defined as 
\begin{equation*}
\mathcal{X}_{i,j}(\hat{z}_j(k|t)) := \left \{\hat{x}_{i}(k|t) : h_{ij}(\hat{x}_{i}(k|t),\hat{z}_j(k|t)) \geq 0 \right \}.
\end{equation*}

\subsection{EKF Predictor}
In order to handle the communication losses, we equip the agents with an Extended Kalman Filter (EKF) as an $N$-step predictor of the future trajectory of the other agents. EKF Predictor is formulated for a general nonlinear system $x^+ = f(x)$ with linearized dynamics around a measured state and unknown input. The $N$-step prediction is done by repeating the prediction step $N$ times in open loop. 

For brevity, we do not state the EKF matrices but refer the reader to the relevant literature \cite{kalman1960new}. The outputs of the EKF Predictor used in the control architecture are the predicted trajectories $\hat{z}_{p,i}(\cdot|t)$ and the prediction covariance matrices $P_{p,i}(\cdot|t)$. The prediction covariance matrix $P_{p,i}(k|t)$ provides us with the estimate of a super ellipsoidal set determined by its eigenvalues $\lambda_k(P_{p,i})$. 

\section{RENDEZVOUS ALGORITHM}\label{sec:rendezvous_algorithm}
The multi-agent formulation in which we have multiple followers is challenging due the possible inter-agent collisions even in the case when all agents share their committed trajectories. In this section, we formulate the algorithm for multiple-follower rendezvous landing.

We assume that the landing platform is large enough to accommodate $M$ agents that can land simultaneously at prespecified positions $c_i\in \mathcal{B}_{r-r_{safe}}$ relative to the center of the landing platform. Moreover, due to the application testbed in hand \cite{andersson2021wara}, we will constrain the scheme to $M$ agents, and additional follower agents can be included if a prioritization procedure is included in the proposed algorithm. We state the following assumption on the initial conditions and feasibility that will be used to analyze the algorithm convergence.

\begin{ass}\label{ass:initial_terminal_conditions}
    All agents at time $t_0$ have initial conditions $x_i(t_0)$, $i\in \mathcal{N}$ such that 
    \begin{equation*}
        x_i(t_0) \in \mathcal{X}_{i,j}(z_j(t_0)) \quad \text{for } i\neq l, j \in \mathcal{N}_{-i}.
    \end{equation*}
    Moreover, it holds that
    \begin{equation*}
        \left \| c_i - c_j \right \| > R, \text{  for all } i,j \in \mathcal{M}, i\neq j.
    \end{equation*}    
\end{ass}

Before we state the algorithm, let us introduce a data collection $D_i(t)$ that is available to an agent $i$ at time step $t$ as $D_i(t) = \{ \hat{z}_{j}(\cdot|t_{j}) \}_{j\in \mathcal{N}_{-i}}$ consisting of the shared future trajectories of all other agents broadcast at time $t_{j}<t, \forall j\in \mathcal{N}_{-i}$.

\begin{figure}[h!]
\vspace{-0.6cm}
\begin{algorithm}[H]
\caption{Multiple-Follower Rendezvous Algorithm}\label{alg:ch6_MF_rendezvous}
\begin{algorithmic}[1]
\Require initial states $x_{i}(0)$ at time $t=0$, landing locations $c_i$ according to Assumption~\ref{ass:initial_terminal_conditions} for $i \in \mathcal{M}$, a tolerance parameter $\varepsilon$;

\For{each agent $i \in \mathcal{M}$}
	\State update $D_i(t)$ and $\hat{z}_j(\cdot|t)$, $j\in \mathcal{N}_{-i}$ using \eqref{eq:ch6_shift_and_predict}
	\State $\hat{x}_i^*(\cdot|t),\hat{u}_i^*(\cdot|t) \gets $ solve Problem~\ref{prob:ch6_docp_standard}
\If{$ \left \| \hat{x}_i(t) - \hat{z}_l(t) + c_i \right \| > \varepsilon$}  
    \State apply $\hat{u}_i^*(0|t)$
    \State broadcast $\hat{x}_i^*(\cdot|t)$
    \State $t \gets t+1$
\EndIf
\EndFor
\end{algorithmic}
\end{algorithm}%
\vspace{-0.8cm}
\end{figure}

\begin{rem}
    An issue that can occur in the multi-agent case with inter-agent collision avoidance in general, is that the agents can end up in a deadlock and be prevented to effectively find a way to navigate to the goal position. In that case, the deadlock can be resolved by forcing the agents to solve Problem~\ref{prob:ch6_docp_standard} sequentially \cite{richards2007robust,christofides2013distributed}. This would guarantee that each agent takes into account the current predicted trajectory of other agents and waits until the process is completed. Thus, the generated trajectories will not end up in a deadlock. Note that communication in this case is required. 
    Therefore, Algorithm~\ref{alg:ch6_MF_rendezvous} requires that the first iteration of the algorithm is done in a sequential manner and that the initial feasibility is established.
\end{rem}

\section{CONVERGENCE}\label{sec:convergence}
Let us consider the one-follower case of the rendezvous landing problem, $\mathcal{N}=\{l,f\}$, in which the leader reference $x_{r,l}(t)$ is given, and the follower reference is based on the position of the leader $x_{r,f}(t) = z_l(t)$ that corresponds to the center of the landing platform. In order to analyze system behavior, we introduce an assumption on the follower's capability to track the leader with respect to the leader dynamics and spatiotemporal constraints.

\begin{ass}\label{ass:ch6_follower_capability}
There exists a control law $\kappa: \mathcal{X}_f \times \mathcal{X}_f \rightarrow \mathcal{U}_f$ such that 
\begin{equation}
\| x_f^+ -z_l^+  \|^2 \leq \rho \| x_f - z_l  \|^2
\end{equation}
with $\rho \in (0,1)$ and
\begin{align*}
x_f^+ &= f_f(x_f,\kappa(x_f,z_l)) \\
x_l^+ &= f_l(x_l,u_l), \\ 
z_l &= Hx_l, \\
h_{C}&(x_f,z_l) \geq 0
\end{align*}
for all $x_f \in \mathcal{X}_f, x_l  \in \mathcal{X}_l, $ and $u_l \in \mathcal{U}_l$.
\end{ass}

Assumption~\ref{ass:ch6_follower_capability} states that for any control action the leader takes, there exists a control law for the follower that will reduce the distance between them in every consecutive time step. 
However, in the case of communication losses, the follower must be capable to asymptotically reduce the distance regardless of the leader's behaviour.

Furthermore, Assumption~\ref{ass:ch6_follower_capability} is similar to Assumptions 4 and 6 in \cite{kohler2018nonlinear} from which the notion of incremental stability as in \cite[Def. 2.1]{angeli2002lyapunov},\cite{tran2016incremental} can be elaborated. Compared to \cite{kohler2018nonlinear}, the reference $z_l$ in our case has different dynamics than the followers' and it does not take into account the control action of the leader.

Referring to Problem~\ref{prob:ch6_docp_standard}, let the value function at time step $t$ be 
\begin{align*}
    V_N(x_f(t),z_l(t)) &= \min_{u_f(\cdot|t)}
J_f(\hat{x}_f(\cdot|t),u_f(\cdot|t),\hat{z}_l(\cdot|t),N,t) \\
&= \sum_{k=0}^{N-1}  \left \| \hat{x}_{f}^*(k|t) - \hat{z}_l(k|t)  \right \|_{Q_{f}}^2
\end{align*}
We define the region of attraction $\mathcal{X}_f^{ROA}(z_l) := \left \{ x_f\in \mathcal{X}_f : V_N(x_f,z_l) \leq V_{N,max} \right \}$ of the MPC controller as the set of states which can be steered to the desired leader state $z_l$ in $N$ or fewer steps.

The convergence result for the case of one follower is based on ensuring that for all initial states in the region of attraction, the value function is a Lyapunov function decreasing at each time step.

\begin{thm}[Convergence with one follower]
Let Assumption~\ref{ass:ch6_follower_capability} hold. For any $V_{N,max} \in \mathbb{R}_{>0}$, there exist constants $\gamma \geq 1$, and $N_0 \in \mathbb{N}$, such that for all $N>N_0$ and all initial conditions in the region of attraction $x_f(0) \in \mathcal{X}_f^{ROA}(z_l(0))$, there exists $\alpha_N \in \mathbb{R}_{>0}$ such that the multi-agent system~\eqref{eq:multi-agent_system}, with  $w_i=0$, $i\in \mathcal{N}=\{l,f\}$ satisfies
\begin{equation*}
\| x_f - z_l \|_{Q_f}^2  \leq V_N(x_f,z_l) \leq \gamma \| x_f - z_l \|_{Q_f}^2 ,
\end{equation*}
\begin{equation*}
V_N(x_f^+,z_l^+) - V_N(x_f,z_l)  \leq - \alpha_N  \| x_f - z_l \|_{Q_f}^2,
\end{equation*}
for all $t\geq 0$. Furthermore, for all $x_f(0) \in \mathcal{X}_f^{ROA}(z_l(0))$ the follower converges to the leader-dependent rendezvous location exponentially, i.e. $x_f(t) \rightarrow z_l(t)$ as $t\rightarrow \infty$.\label{thm:ch6_convergence_one_follower}
\end{thm}
The proof is given in Appendix~\ref{sec:ch6_app_thm1}.

\begin{rem}
Note that the optimization problem in Prob.~\ref{prob:ch6_docp_standard} does not use the terminal ingredients and thus they are not used in Theorem~\ref{thm:ch6_convergence_one_follower}, although it is a common way to prove the stability of MPC scheme \cite{mayne2000constrained,chen1998quasi}. 
In this work, we avoid usage of the terminal ingredients by considering a sufficiently long planning horizon $N$ in the region of attraction $\mathcal{X}^{ROA}$. 
This builds upon the methodology suggested in \cite{boccia2014stability,kohler2018nonlinear}.
Moreover, the initial feasibility in the region of attraction is implicitly assumed in Theorem~\ref{thm:ch6_convergence_one_follower} by the same principle, i.e. by assuming that there exists $N_0$ for which the optimization problem in Prob.~\ref{prob:ch6_docp_standard} is feasible for all initial conditions in $\mathcal{X}^{ROA}$. Also note that the disturbances are not considered in the theorem.
\end{rem}
It is straightforward to show, using the same argument as in Theorem~\ref{thm:ch6_convergence_one_follower}, that the leader agent will follow its reference trajectory as well, thus achieving the objective of Problem~\ref{prob:formulation}.

\subsection{Robustness to Communication Losses}
In case of communication loss, the latest available shared trajectory of the leader used as the follower reference is shifted and the missing part of the trajectory is predicted with the EKF predictor. If the latest time of arrival $t_a$ of the shared trajectory is $t_a < t-k$, $0 < k < N$, then the trajectory is shifted for $k$ time steps and the rest is predicted as follows 
\begin{equation}
	\hat{z}_{i}(l|t) = \begin{cases}
	\hat{z}_{i}(l+k|t_{a,i}), &\text{for } l<N-k \\
	\hat{z}_{p,i}(l+k|t_{a,i}), &\text{for } l\geq N-k
	\end{cases}\label{eq:ch6_shift_and_predict}
\end{equation} 
for $l=0,1,...,N$. 
However, the predicted trajectory has some uncertainty such that $z_{i}(l|t) \in \hat{z}_{p,i}(l|t_i) \oplus P_i(l|t)$ where the set $P_i(l|t):=\{ b :  \| b \|_{P_{p,i}(l|t)}^2 \leq s \}$ is determined by the prediction covariance matrix $P_{p,i}(l|t)$ and a parameter $s=-2\ln(1-p)$, that depends on the chosen probability confidence $p\in(0,1)$. 
In case $k = N-1$ which means that no predicted future steps are available, the data collection is updated only with the EKF Predictor. 
\begin{rem}

Note that the next state uncertainty can be estimated without using EKF Predictor as $z_{i}(l+1|t) \in \hat{z}_{i}(l|t_i) \oplus \mathcal{B}_{r_i}$ where the choice of the safety radius $r_i$ determines the conservativeness of the used set estimates. By setting $r_i$  to
\begin{equation}
r_i = \max_{x_i \in \mathcal{X}_i,u_i \in \mathcal{U}_i} \| x_i^+ - f(x_i,u_i) \| \label{eq:ch6_worst_case_radius}
\end{equation}
one can guarantee that the next state is within the ball of the given radius. However, this is an overly conservative approach given that the uncertainty sets grow and can become very large at the end of the horizon thus preventing the follower agents to land. In that case, the collision checking can be restricted only for the one-step ahead prediction, i.e.
\begin{equation*}
\left \| C(\hat{x}_i(1|t) - \hat{x}_j(1|t)) \right \| \geq R + r_j, \quad \text{for } j\in\mathcal{N}_{-i}
\end{equation*} 
and use the worst case radius $r_j$ as in \eqref{eq:ch6_worst_case_radius}.
\end{rem}
\begin{thm}[$p$-probabilistically safe landing]\label{thm:probabilistic_landing}
    Let the conditions of Theorem~\ref{thm:ch6_convergence_one_follower} hold. Given the probabilistic confidence $p\in(0,1)$, the radius of the landing platform $r$, and the radius necessary for the safe landing $r_{safe}$, if the following condition holds
    \begin{equation*}
        r_{safe} + \sqrt{s\lambda_{\max}(P_{p,i}(t))} < r
    \end{equation*}
    where $s=-2\ln(1-p)$, then the landing is considered as probabilistically safe with probability $p$.
\end{thm}
\begin{proof}
    The proof is based on the worst-case estimate of the landing position. Given the probabilistic confidence $p$ and the covariance matrix $P_{p,i}(t)$, the worst-case distance from the actual landing position $z_l(t)$ and its estimate $\hat{z}_l(t)$ is $d=\sqrt{s\lambda_{\max}(P_{p,i}(t))}$. Thus, $\|z_l(t) - \hat{z}_l(t)\| = d < r-r_{safe}$ which means that the follower applying the control input obtained with Prob.~\ref{prob:ch6_docp_standard} is guaranteed to land inside of the landing platform with probability $p$. 
\end{proof}

The extension to the multiple-follower case is given by the following result:
\begin{cor}\label{thm:rendezvous_algorithm}
    Let Assumption~\ref{ass:ch6_follower_capability} and~\ref{ass:initial_terminal_conditions} hold. Moreover, let the conditions of Theorem~\ref{thm:probabilistic_landing} hold for all follower agents $i\in\mathcal{M}$. Then, Algorithm~\ref{alg:ch6_MF_rendezvous} converges and all follower agents meet on the leader landing platform without a collision.
\end{cor}
\begin{proof}
    The proof is based on two parts. First, that the all agents converge to the landing platform and second, that their trajectories are collision-free.  
    Given that all stated conditions hold there exists a feasible configuration for all follower agents to rendezvous on the leader landing platform. Moreover, all follower agents also satisfy the conditions from Theorem~\ref{thm:probabilistic_landing}, and thus there exist a feasible landing trajectory robust to communication losses.
    Because each agent is solving the optimization problem in Problem~\ref{prob:ch6_docp_standard}, the inter-agent collision avoidance is enforced in every feasible landing trajectory. Therefore, all agents rendezvous on the leader landing platform without a collision. 
\end{proof}

\section{SIMULATION RESULTS}\label{sec:sim_results}
In this section we present a landing scenario with $M=6$ agents. The leader is unable to communicate with the follower agents and thus follower agents must use EKF Predictor to estimate the position of the leader. Moreover, the leader measurements are taken from real-world experiments and thus have disturbances. The follower agents share their predicted trajectories until time step $k=10$ (continuous time $t_c = 2 s$) when one of the agents (Agent $f_1$) also loses the communication with the rest of the agents. Then the rest of the agents in the scenario must also predict the future trajectory of Agent $f_1$ and Agent $f_1$ predicts the trajectories of all other agents in the scenario.

The state and input constraints on the models are defined similarly to \cite[Sec.5.1]{lapandic2021aperiodic}. The initial positions of follower agents are $x_i(0) = [5\cos(2i\pi/M),5\sin(2i\pi/M),10, 0_6^T]^T$, and the leader is at origin $x_l(0)=0_6$. The radius of the whole landing platform for all agents is $5r_{safe}$, and the safe radius for landing is $r_{safe}=0.5m$. The matrices $Q_{f_i}=\text{diag}(10,10,5,1,1,1,1,1,1)$, $Q_{l}=\text{diag}(10,10,10,1,1,1)$, thus primary penalizing the position in $x$ and $y$ and then in $z$ for quadrotors, and orientation $\psi$ for the boat. $\lambda_{\max}(Q_{f_i})=10$, and we pick $V_{N,max}=240$ such that $\bar{\gamma}=1.99$ and $N=20$, and the region of attraction is large enough and encompass the initial displacement with a sufficient margin.

The landing locations are equidistantly distributed as on Fig.~\ref{fig:landing_platform} and assigned to an agent positioned at the opposite side diagonally. Assumption~\ref{ass:initial_terminal_conditions} is satisfied with $R=2r_{safe}$, and $C=I_3$. 
The landing is considered safe if the conditions of Theorem~\ref{thm:probabilistic_landing} hold with $p=0.95$. Assuming that all agents will have identical estimation of the landing platform position and inter-agent collisions are handled it is sufficient to consider the safety with respect to the outer boundary. Thus, with $p=0.95$, $\lambda_{max}(P_{p,l}(t)) < (5r_{safe}-3r_{safe})^2/s \approx 0.17$. From the experimental results, $\lambda_{max}(P_{p,l}(t))<0.06$ thus satisfying Theorem~\ref{thm:probabilistic_landing}. Algorithm~\ref{alg:ch6_MF_rendezvous} and Prob.~\ref{prob:ch6_docp_standard} are implemented with CasADi \cite{Andersson2019casadi} and results are shown on Figures~\ref{fig:big_fig} and \ref{fig:top_view}. All trajectories are collision-free, and the predicted trajectory of Agent $f_1$ by other agents and vice versa do not induce much conservativeness to Algorithm. This is mainly because the first part of the trajectory is generated using the shift mechanism as in \eqref{eq:ch6_shift_and_predict} and the small eigenvalues of the covariance matrices compared to the considered safety radii. 

\begin{figure}
   \centering
        \includegraphics[width=0.4\linewidth]{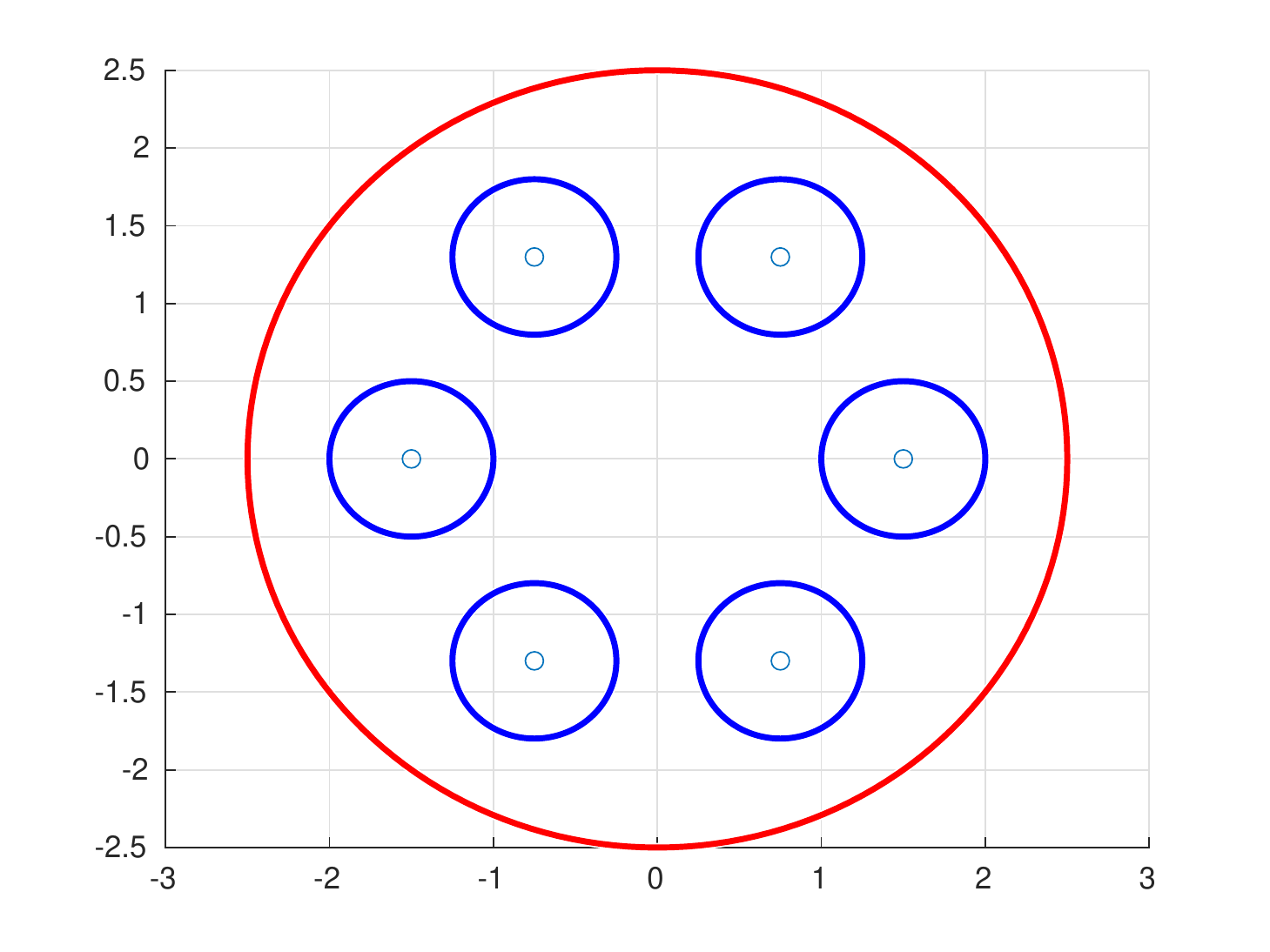} 
        \caption{The red circle marks the boundary of the landing platform, while the blue circles are the landing locations with safety radius $r_{safe}=0.5m$. The distance from the center of the platform to a landing location is $3r_{safe}.$}\label{fig:landing_platform}
    \vspace{-0.6cm}
\end{figure}

\begin{figure*}[t]
    \centering
    {\includegraphics[width=0.3\textwidth]{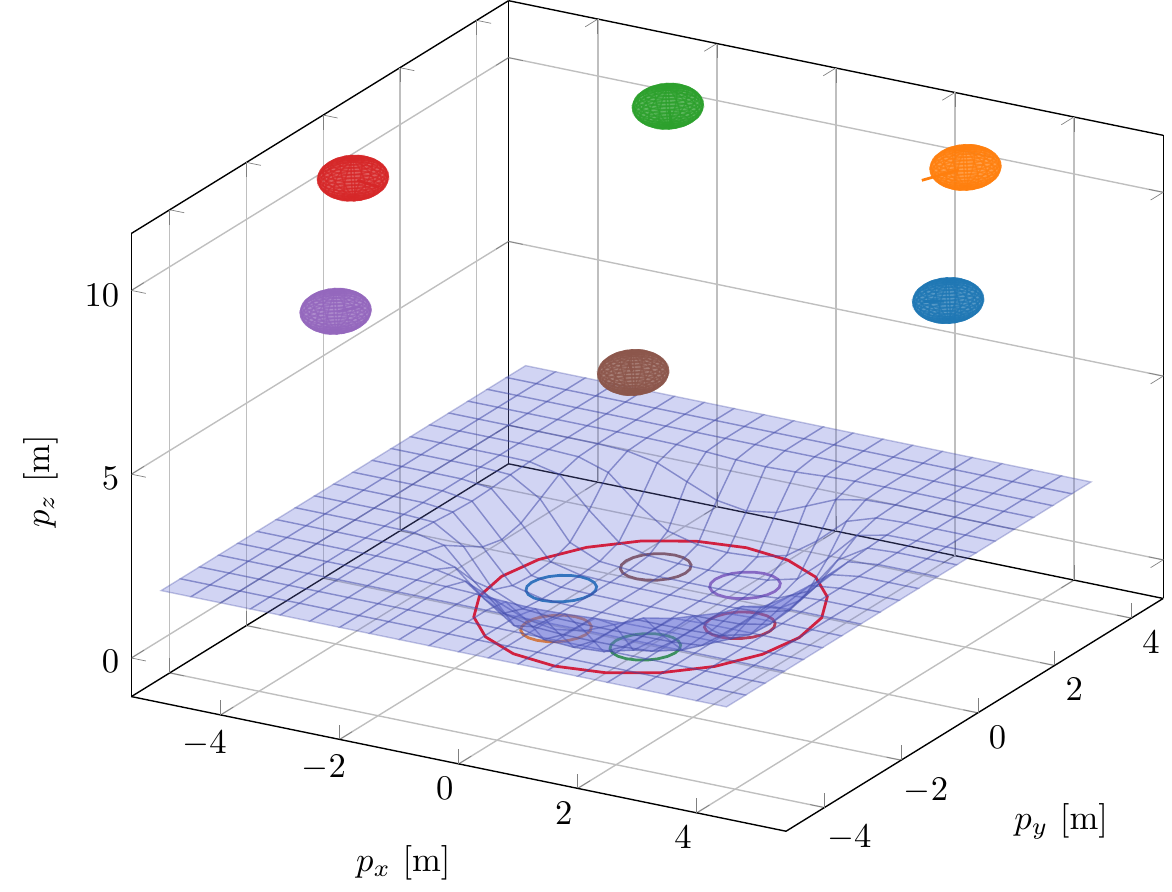}}\hfill
    {\includegraphics[width=0.3\textwidth]{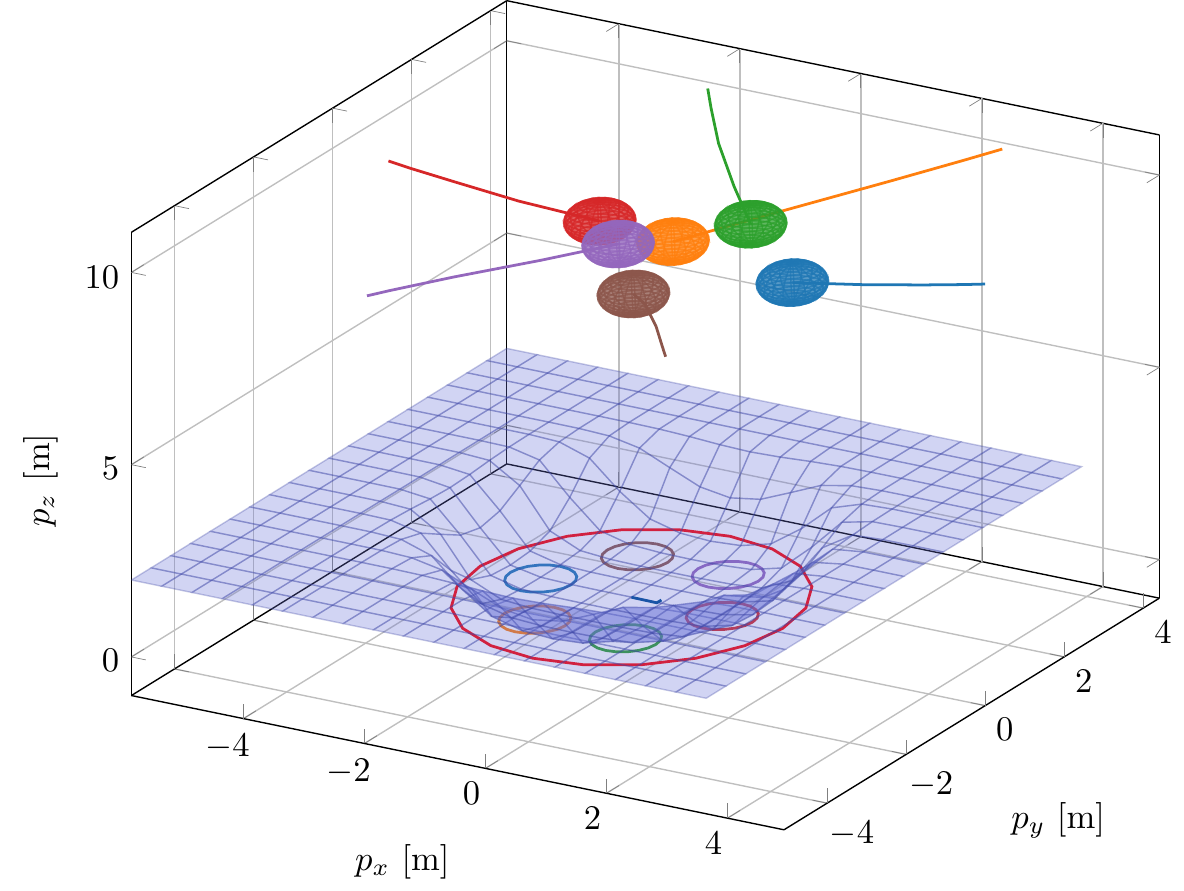}}\hfill
    {\includegraphics[width=0.3\textwidth]{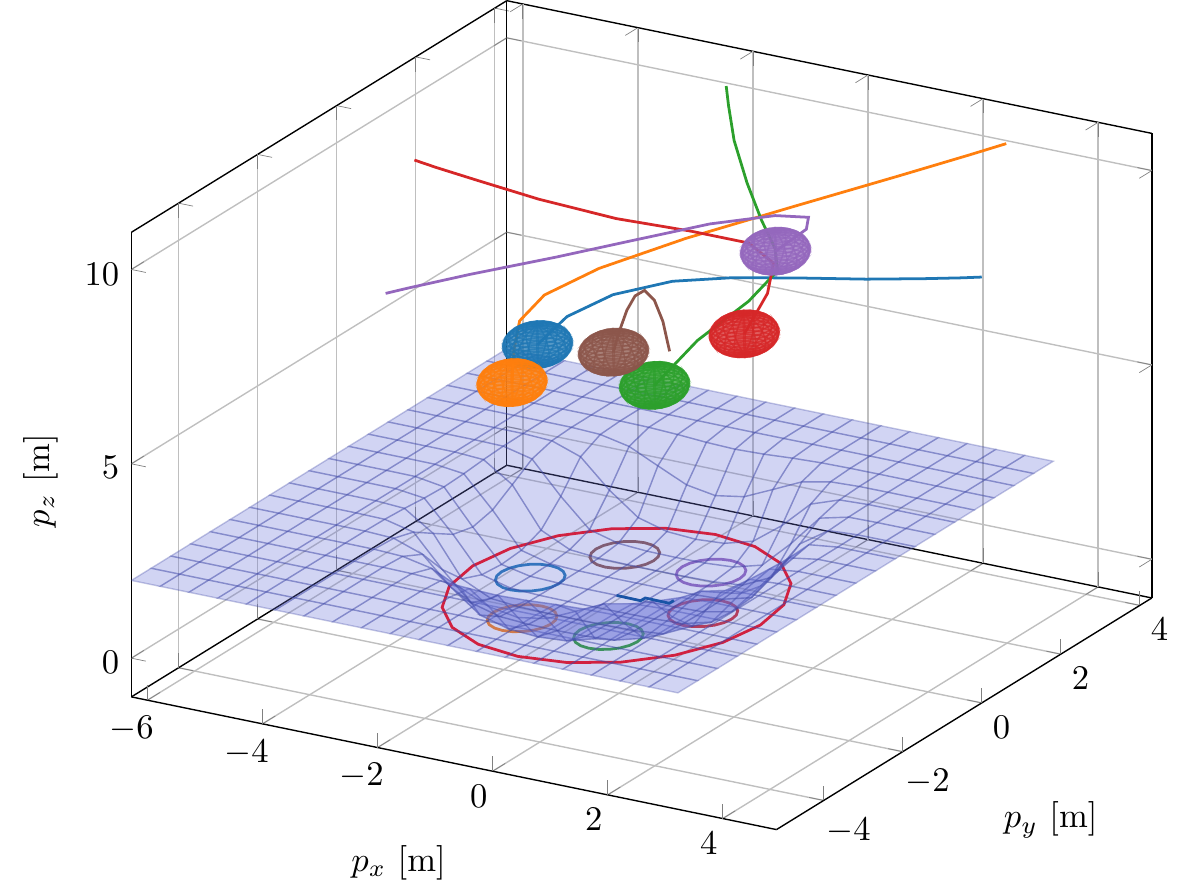}}\hfill
    {\includegraphics[width=0.3\textwidth]{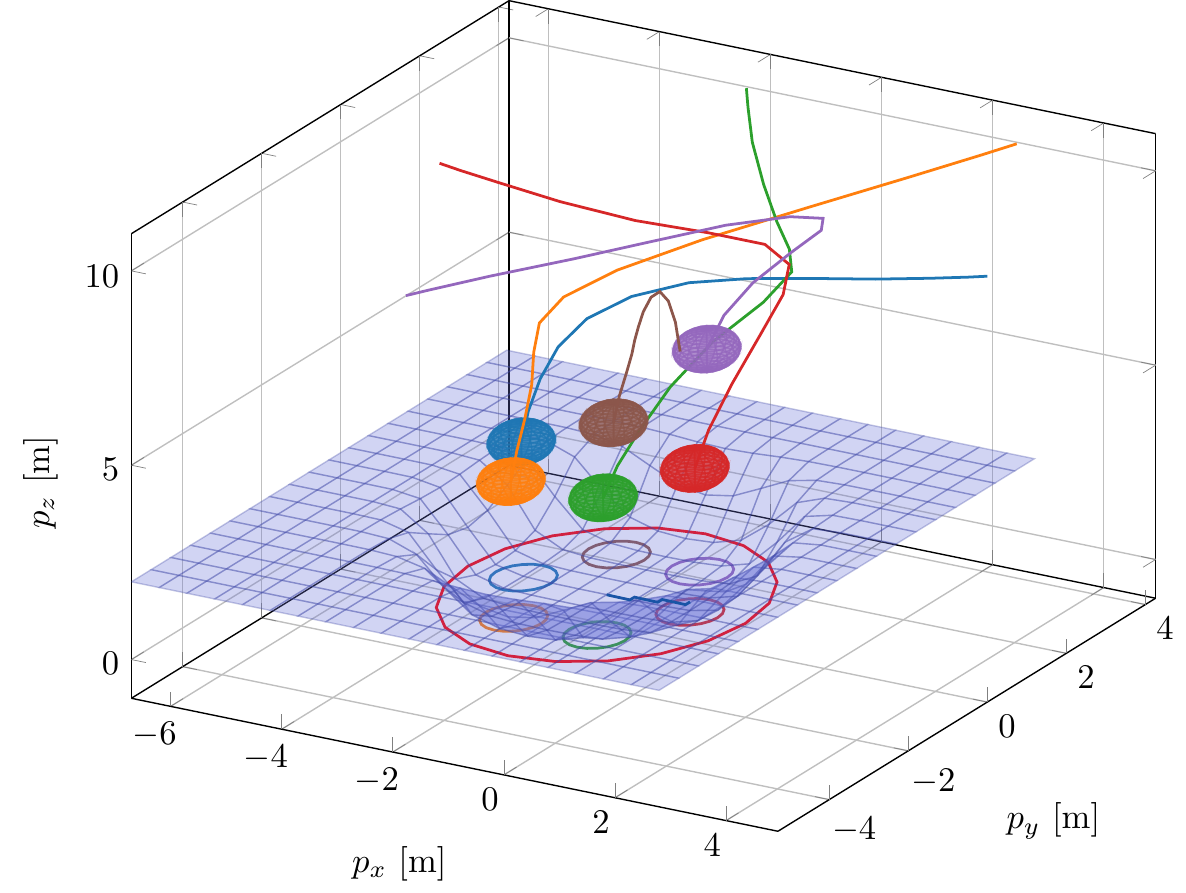}}\hfill
    {\includegraphics[width=0.3\textwidth]{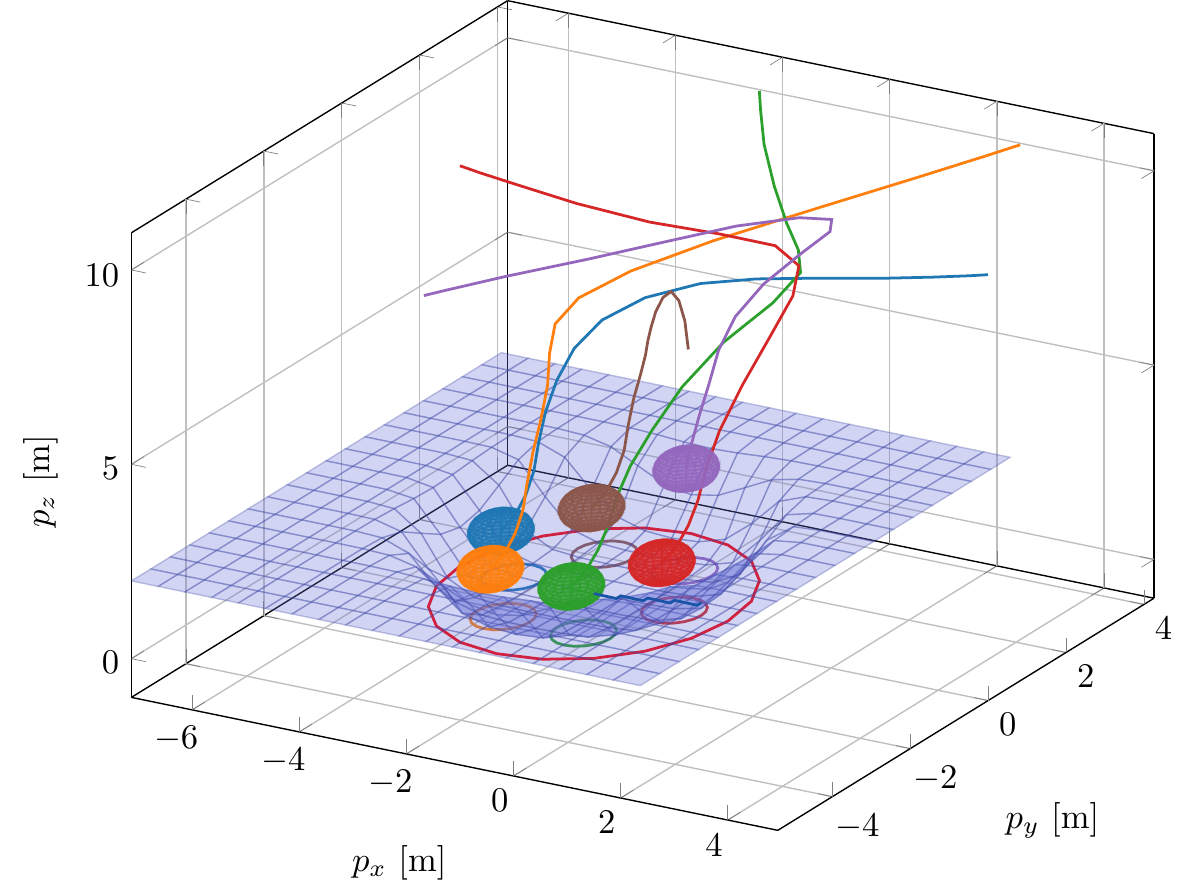}}\hfill
    {\includegraphics[width=0.3\textwidth]{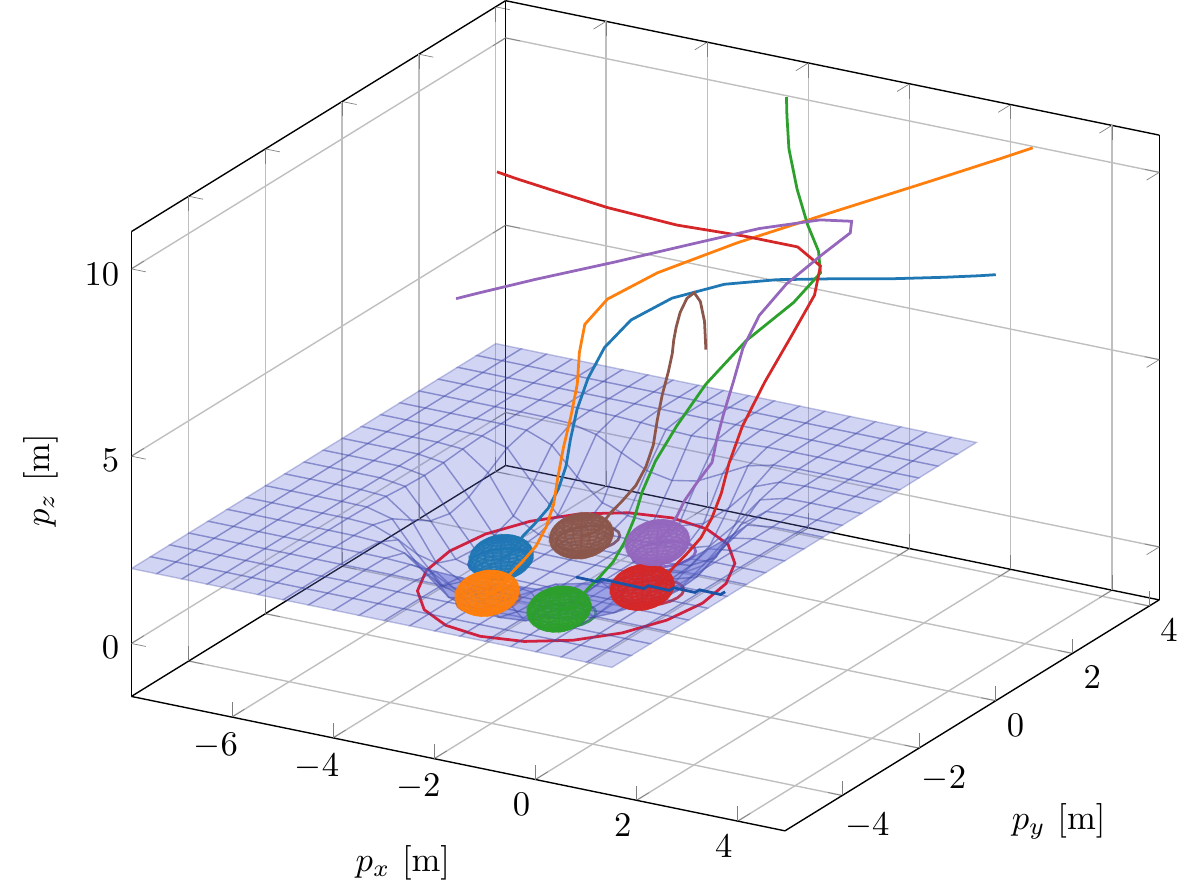}}
    \caption{3D view of a landing scenario at discrete time steps $t=\{0,5,10,13,17,23 \}$ with $\Delta t=0.2s$ between two time steps. The followers are unable to communicate with the leader and a loss of communication with Agent $f_1$ occurs during the experiment at $t=10$. Agent $f_1$ is marked in orange.}\label{fig:big_fig}
\end{figure*}
\begin{figure}
        \centering
        {\includegraphics[width=0.47\linewidth]{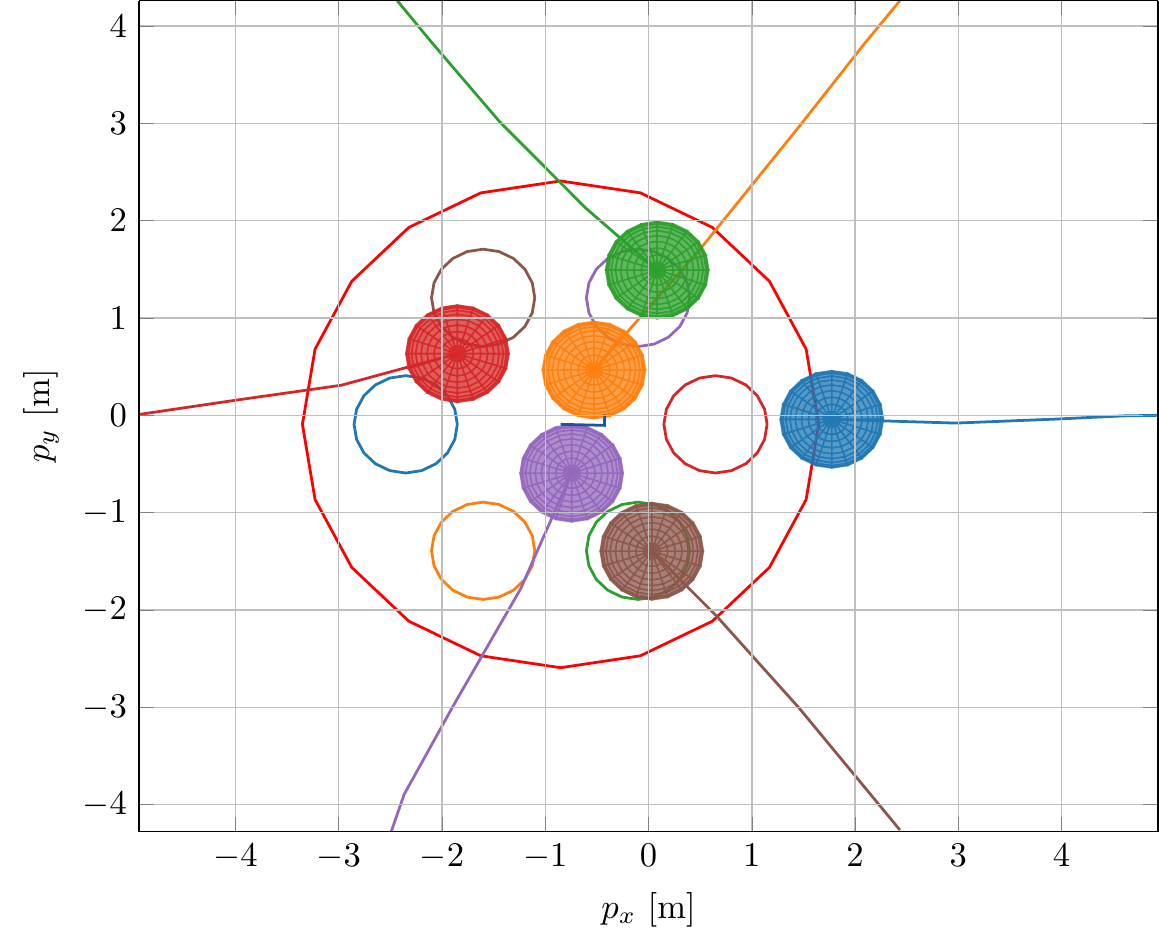}} \hfill 
        {
        \includegraphics[width=0.47\linewidth]{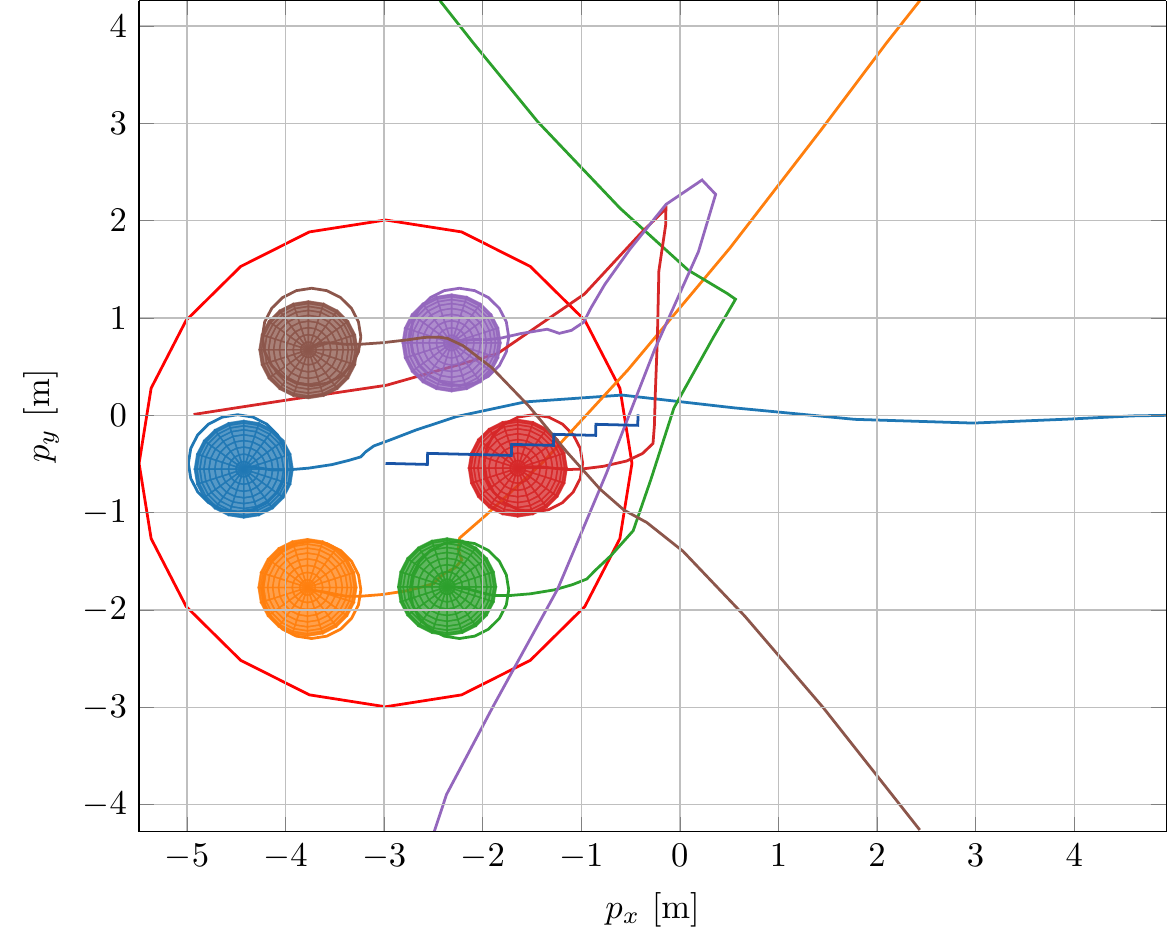}
        }
        \caption{Left: Top view of the second situation at $t=5$, when it might look from the perspective view that the agents are too close, shows that the inter-agent collision avoidance constraint is enforced. Right: Final situation at $t=23$. $h_C$ is removed for better visibility.}
        \label{fig:top_view}
    \vspace{-0.52cm}
\end{figure}

\section{CONCLUSION}\label{sec:conclusion}
In this paper, we presented a rendezvous algorithm based on the leader-follower scheme and distributed MPC with robustness to communication losses. The algorithm is designed for autonomous landing of multiple quadrotors on moving unmanned surface vehicles. The convergence analysis of the algorithm is presented and the effectiveness of the proposed algorithm is demonstrated with the simulation of a landing scenario.In the future work, we aim to include the disturbances in the analysis and quantify the upper bounds such that the convergence is preserved. Moreover, it will be interesting to include learning methods in order to speed-up the computation.

\addtolength{\textheight}{-2cm}   



\begin{appendices}
\renewcommand{\thesection}{\Alph{section}}
\renewcommand{\thesectiondis}[2]{\Alph{section}:}

\section{Proof of Theorem~\ref{thm:ch6_convergence_one_follower}}
\label{sec:ch6_app_thm1}
\begin{proof}
The proof proceeds in two parts. The first part shows the boundedness of $V_N(x_f(t),z_l(t))$ for all $t\geq 0$, and the second part proves that the value function decreases for all states in the region of attraction at every time step.

By the definition of the value function 
\begin{align*}
V_N(x_f(t),z_l(t)) 
&\geq \left \| \hat{x}_f(0|t) - \hat{z}_l(0|t) \right \|_{Q_f}^2 \\
&= \left \| x_f(t) - z_l(t) \right \|_{Q_f}^2 
\end{align*}
Moreover, for $x_f(t) \in \mathcal{X}_f^{ROA}(z_l)$
\begin{equation*}
\left \| x_f(t) - z_l(t) \right \|_{Q_f}^2  \leq V_N(x_f(t),z_l(t)) \leq V_{N,max}
\end{equation*}
thus, there exists $\bar{\gamma} \geq 1$ for which $
\bar{\gamma}\left \| x_f(t) - z_l(t) \right \|_{Q_f}^2  = V_{N,max}$ and $\gamma \geq \bar{\gamma} \geq 1$ such that
\begin{equation*}
 V_N(x_f(t),z_l(t)) \leq V_{N,max} \leq \gamma \| x_f(t) - z_l(t) \|_{Q_f}^2 
\end{equation*}

The second part of the proof uses the contraction of the error from Assumption~\ref{ass:ch6_follower_capability}. 
Let a feasible (suboptimal) input sequence for the next time step be $\tilde{u}_f(\cdot|t+1)$, defined as
\begin{equation*}
   \tilde{u}_f(k|t+1) = \begin{cases}
       u_f^*(k+1|t) &\text{for } k=0,1,...,N-1 \\
       \kappa(\hat{x}_f^*(N|t),\hat{z}_l(N|t))  &\text{for } k=N
   \end{cases}
\end{equation*}
consisting of the shifted optimal input from the previous step and some appended $u_f^N = \kappa(\hat{x}_f^*(N|t),\hat{z}_l(N|t)) \in \mathcal{U}_f$ that satisfies Assumption~\ref{ass:ch6_follower_capability}.
Assuming there are no disturbances, then $
    x_f(t+1) = \tilde{x}_f(0|t+1) = f_f(x_f(t),u_f^*(0|t))
$
and 
\begin{equation*}
   \tilde{x}_f(k|t+1) = \begin{cases}
       \hat{x}_f^*(k+1|t) &\text{for } k<N \\
       f_f(\hat{x}_f^*(N|t),u_f^N) &\text{for } k=N
   \end{cases}
\end{equation*}

Moreover, due to Assumption~\ref{ass:ch6_follower_capability}
\begin{equation*}
   \left \| \tilde{x}_{f}(N|t+1) - \hat{z}_l(N|t+1)  \right \|^2 \leq \rho \left \| \hat{x}_f^*(N|t) - \hat{z}_l(N|t)  \right \|^2
\end{equation*}
and
\begin{align*}
    \| \tilde{x}_{f}(N|t+1) &- \hat{z}_l(N|t+1)   \|_{Q_{f}}^2 \\
   &\leq \rho \frac{\lambda_{max}(Q_f)}{\lambda_{min}(Q_f)} \left \| \hat{x}_f^*(N|t) - \hat{z}_l(N|t)  \right \|_{Q_{f}}^2
\end{align*}
Because $ V_N(x_f(t),z_l(t)) \leq \gamma \| x_f(t) - z_l(t) \|_{Q_f}^2 $, then $\forall~k~\geq~1$, there exist $N>1$ and $\gamma \geq \bar{\gamma} \geq 1$ such that 
\begin{align*}
    \| \hat{x}_f^*(k|t) - \hat{z}_l(k|t) \|_{Q_f}^2 &\leq \frac{V_N(x_f(t),z_l(t))}{N} \\
    &\leq \frac{\gamma}{N}\| x_f(t) - z_l(t) \|_{Q_f}^2
\end{align*}
Let us consider the the value function at time step $t+1$, then
\begin{align*}
V_N&(x_f(t+1),z_l(t+1)) \\
&\leq \sum_{k=0}^{N}  \left \| \tilde{x}_{f}(k|t+1) - \hat{z}_l(k|t+1)  \right \|_{Q_{f}}^2 \\
&= \sum_{k=0}^{N-1}  \left \| \tilde{x}_{f}(k|t+1) - \hat{z}_l(k|t+1)  \right \|_{Q_{f}}^2 \\
&\phantom{=}+ \left \| \tilde{x}_{f}(N|t+1) - \hat{z}_l(N|t+1)  \right \|_{Q_{f}}^2 \\
&= \sum_{k=1}^{N}  \left \| \hat{x}_{f}^*(k|t) - \hat{z}_l(k|t)  \right \|_{Q_{f}}^2 \\
&\phantom{=}+ \left \| \tilde{x}_{f}(N|t+1) - \hat{z}_l(N|t+1)  \right \|_{Q_{f}}^2 
\end{align*}
\begin{align*}
&= \sum_{k=0}^{N}  \left \| \hat{x}_{f}^*(k|t) - \hat{z}_l(k|t)  \right \|_{Q_{f}}^2 - \left \| \hat{x}_{f}^*(0|t) - \hat{z}_l(0|t)  \right \|_{Q_{f}}^2 \\
&\phantom{=}+ \left \| \tilde{x}_{f}(N|t+1) - \hat{z}_l(N|t+1)  \right \|_{Q_{f}}^2 \\
&= V_N(x_f(t),z_l(t)) - \left \| x_{f}(t) - z_l(t)  \right \|_{Q_{f}}^2 \\
&\phantom{=}+ \left \| \tilde{x}_{f}(N|t+1) - \hat{z}_l(N|t+1)  \right \|_{Q_{f}}^2 \\
&\leq V_N(x_f(t),z_l(t)) - \left \| x_{f}(t) - z_l(t)  \right \|_{Q_{f}}^2 \\
&\phantom{=}+ \rho\frac{\gamma}{N} \frac{\lambda_{max}(Q_f)}{\lambda_{min}(Q_f)} \left \| x_{f}(t) - z_l(t)  \right \|_{Q_{f}}^2 \\
&\leq V_N(x_f(t),z_l(t)) - \alpha_N \left \| x_{f}(t) - z_l(t)  \right \|_{Q_{f}}^2
\end{align*}
where $\alpha_N:= 1- \rho \frac{\gamma}{N} \frac{\lambda_{max}(Q_f)}{\lambda_{min}(Q_f)}$. Thus by choosing $N > N_0:= \bar{\gamma} \frac{\lambda_{max}(Q_f)}{\lambda_{min}(Q_f)}$, $\alpha_N > 0$. 
Finally, the value function $V_N(x_f(t),z_l(t))$ is decreasing for all $t\geq 0$. 

Using the decrease property and the boundedness in the region of attraction $\mathcal{X}_f^{ROA}(z_l)$ proven in the first part, $V_N(x_f(t),z_l(t))$ is a Lyapunov function in $\mathcal{X}_f^{ROA}(z_l)$.  Thus, the error $\left \| x_{f}(t) - z_l(t)  \right \|$ exponentially goes to zero, which concludes the proof. 
\vspace{-0.25cm}
\end{proof}
\end{appendices}

\bibliographystyle{unsrt}
\bibliography{root}

\begin{thebibliography}{10}

\bibitem{muller2012cooperative}
Matthias~A M{\"u}ller, Marcus Reble, and Frank Allg{\"o}wer.
\newblock Cooperative control of dynamically decoupled systems via distributed
  model predictive control.
\newblock {\em International Journal of Robust and Nonlinear Control},
  22(12):1376--1397, 2012.

\bibitem{christofides2013distributed}
Panagiotis~D Christofides, Riccardo Scattolini, David~Munoz de~la Pena, and
  Jinfeng Liu.
\newblock Distributed model predictive control: A tutorial review and future
  research directions.
\newblock {\em Computers \& Chemical Engineering}, 51:21--41, 2013.

\bibitem{verginis2018communication}
Christos~K Verginis, Alexandros Nikou, and Dimos~V Dimarogonas.
\newblock Communication-based decentralized cooperative object transportation
  using nonlinear model predictive control.
\newblock In {\em 2018 European control conference (ECC)}, pages 733--738.
  IEEE, 2018.

\bibitem{lopez2018robust}
Brett~T Lopez, Jean-Jacques Slotine, and Jonathan~P How.
\newblock Robust collision avoidance via sliding control.
\newblock In {\em 2018 IEEE International Conference on Robotics and Automation
  (ICRA)}, pages 2962--2969. IEEE, 2018.

\bibitem{csomay2022multi}
Noel Csomay-Shanklin, Andrew~J Taylor, Ugo Rosolia, and Aaron~D Ames.
\newblock Multi-rate planning and control of uncertain nonlinear systems: Model
  predictive control and control lyapunov functions.
\newblock In {\em 2022 IEEE 61st Conference on Decision and Control (CDC)},
  pages 3732--3739. IEEE, 2022.

\bibitem{rosolia2022unified}
Ugo Rosolia, Andrew Singletary, and Aaron~D Ames.
\newblock Unified multirate control: From low-level actuation to high-level
  planning.
\newblock {\em IEEE Transactions on Automatic Control}, 67(12):6627--6640,
  2022.

\bibitem{falanga2017vision}
Davide Falanga, Alessio Zanchettin, Alessandro Simovic, Jeffrey Delmerico, and
  Davide Scaramuzza.
\newblock Vision-based autonomous quadrotor landing on a moving platform.
\newblock In {\em 2017 IEEE International Symposium on Safety, Security and
  Rescue Robotics (SSRR)}, pages 200--207. IEEE, 2017.

\bibitem{song2022distributed}
Yeongho Song, Hojin Lee, Cheolhyeon Kwon, Hyo-Sang Shin, and Hyondong Oh.
\newblock Distributed estimation of stochastic multiagent systems for
  cooperative control with a virtual network.
\newblock {\em IEEE Transactions on Systems, Man, and Cybernetics: Systems},
  2022.

\bibitem{lymperopoulos2008adaptive}
Ioannis Lymperopoulos and John Lygeros.
\newblock Adaptive aircraft trajectory prediction using particle filters.
\newblock In {\em AIAA Guidance, Navigation and Control Conference and
  Exhibit}, page 7387, 2008.

\bibitem{maniatopoulos2012decentralized}
Spyros Maniatopoulos, Dimos~V Dimarogonas, and Kostas~J Kyriakopoulos.
\newblock A decentralized event-based predictive navigation scheme for
  air-traffic control.
\newblock In {\em 2012 American Control Conference (ACC)}, pages 2503--2508.
  IEEE, 2012.

\bibitem{felsen2018will}
Panna Felsen, Patrick Lucey, and Sujoy Ganguly.
\newblock Where will they go? predicting fine-grained adversarial multi-agent
  motion using conditional variational autoencoders.
\newblock In {\em Proceedings of the European conference on computer vision
  (ECCV)}, pages 732--747, 2018.

\bibitem{li2019interaction}
Jiachen Li, Hengbo Ma, and Masayoshi Tomizuka.
\newblock Interaction-aware multi-agent tracking and probabilistic behavior
  prediction via adversarial learning.
\newblock In {\em 2019 international conference on robotics and automation
  (ICRA)}, pages 6658--6664. IEEE, 2019.

\bibitem{zhao2019multi}
Tianyang Zhao, Yifei Xu, Mathew Monfort, Wongun Choi, Chris Baker, Yibiao Zhao,
  Yizhou Wang, and Ying~Nian Wu.
\newblock Multi-agent tensor fusion for contextual trajectory prediction.
\newblock In {\em Proceedings of the IEEE/CVF Conference on Computer Vision and
  Pattern Recognition}, pages 12126--12134, 2019.

\bibitem{li2020evolvegraph}
Jiachen Li, Fan Yang, Masayoshi Tomizuka, and Chiho Choi.
\newblock Evolvegraph: Multi-agent trajectory prediction with dynamic
  relational reasoning.
\newblock {\em Advances in neural information processing systems},
  33:19783--19794, 2020.

\bibitem{cao2022advdo}
Yulong Cao, Chaowei Xiao, Anima Anandkumar, Danfei Xu, and Marco Pavone.
\newblock Advdo: Realistic adversarial attacks for trajectory prediction.
\newblock In {\em European Conference on Computer Vision}, pages 36--52.
  Springer, 2022.

\bibitem{lapandic2022robust}
D{\v{z}}enan Lapandi{\'c}, Christos~K Verginis, Dimos~V Dimarogonas, and
  Bo~Wahlberg.
\newblock Robust trajectory tracking for underactuated quadrotors with
  prescribed performance.
\newblock In {\em 2022 IEEE 61st Conference on Decision and Control (CDC)},
  pages 3351--3358. IEEE, 2022.

\bibitem{fossen2011handbook}
Thor~I Fossen.
\newblock {\em Handbook of marine craft hydrodynamics and motion control}.
\newblock John Wiley \& Sons, 2011.

\bibitem{persson2019model}
Linnea Persson and Bo~Wahlberg.
\newblock Model predictive control for autonomous ship landing in a search and
  rescue scenario.
\newblock In {\em AIAA Scitech 2019 Forum}, page 1169, 2019.

\bibitem{kalman1960new}
Rudolph~Emil Kalman.
\newblock A new approach to linear filtering and prediction problems.
\newblock 1960.

\bibitem{andersson2021wara}
Olov Andersson, Patrick Doherty, M{\aa}rten Lager, Jens-Olof Lindh, Linnea
  Persson, Elin~A Topp, Jesper Tordenlid, and Bo~Wahlberg.
\newblock {WARA-PS}: a research arena for public safety demonstrations and
  autonomous collaborative rescue robotics experimentation.
\newblock {\em Autonomous Intelligent Systems}, 1(1):1--31, 2021.

\bibitem{richards2007robust}
Arthur Richards and Jonathan~P How.
\newblock Robust distributed model predictive control.
\newblock {\em International Journal of control}, 80(9):1517--1531, 2007.

\bibitem{kohler2018nonlinear}
Johannes K{\"o}hler, Matthias~A M{\"u}ller, and Frank Allg{\"o}wer.
\newblock Nonlinear reference tracking with model predictive control: An
  intuitive approach.
\newblock In {\em 2018 European Control Conference (ECC)}, pages 1355--1360.
  IEEE, 2018.

\bibitem{angeli2002lyapunov}
David Angeli.
\newblock A lyapunov approach to incremental stability properties.
\newblock {\em IEEE Transactions on Automatic Control}, 47(3):410--421, 2002.

\bibitem{tran2016incremental}
Duc~N Tran, Bj{\"o}rn~S R{\"u}ffer, and Christopher~M Kellett.
\newblock Incremental stability properties for discrete-time systems.
\newblock In {\em 2016 IEEE 55th Conference on Decision and Control (CDC)},
  pages 477--482. IEEE.

\bibitem{mayne2000constrained}
David~Q Mayne, James~B Rawlings, Christopher~V Rao, and Pierre~OM Scokaert.
\newblock Constrained model predictive control: Stability and optimality.
\newblock {\em Automatica}, 36(6):789--814, 2000.

\bibitem{chen1998quasi}
Hong Chen and Frank Allg{\"o}wer.
\newblock A quasi-infinite horizon nonlinear model predictive control scheme
  with guaranteed stability.
\newblock {\em Automatica}, 34(10):1205--1217, 1998.

\bibitem{boccia2014stability}
Andrea Boccia, Lars Gr{\"u}ne, and Karl Worthmann.
\newblock Stability and feasibility of state constrained mpc without
  stabilizing terminal constraints.
\newblock {\em Systems \& control letters}, 72:14--21, 2014.

\bibitem{lapandic2021aperiodic}
Dženan Lapandić, Linnea Persson, Dimos~V. Dimarogonas, and Bo~Wahlberg.
\newblock Aperiodic communication for mpc in autonomous cooperative landing.
\newblock {\em IFAC-PapersOnLine}, 54(6):113--118, 2021.
\newblock 7th IFAC Conference on Nonlinear Model Predictive Control NMPC 2021.

\bibitem{Andersson2019casadi}
Joel A~E Andersson, Joris Gillis, Greg Horn, James~B Rawlings, and Moritz
  Diehl.
\newblock {CasADi} -- {A} software framework for nonlinear optimization and
  optimal control.
\newblock {\em Mathematical Programming Computation}, 11(1):1--36, 2019.

\end{thebibliography}

\end{document}